\documentclass[aps,twocolumn,prl,floatfix,superscriptaddress,8pt,tightenlines]{revtex4-2}
\RequirePackage{float}
\usepackage[utf8]{inputenc}

\usepackage[USenglish]{babel}
\usepackage[T1]{fontenc}
\usepackage{graphicx,epsf}
\usepackage{subfigure}
\usepackage[toc,page,header]{appendix}
\usepackage{minitoc}
\usepackage{color}
\usepackage{breakurl}
\usepackage[bookmarksnumbered,hypertexnames=false,colorlinks,colorlinks=true,linkcolor=blue,urlcolor=black,citecolor=blue,anchorcolor=green,breaklinks]{hyperref}
\usepackage{bbm,braket,microtype,mathrsfs,amsmath,amssymb,color,amsthm,mathtools,fullpage,graphicx,enumitem,ae,aecompl,bm,thm-restate}
\usepackage{tikz}
\makeatletter

\allowdisplaybreaks
\newcommand{\be}{\begin{equation}}
\newcommand{\ee}{\end{equation}}
\newcommand{\ba}{\begin{eqnarray}}
\newcommand{\ea}{\end{eqnarray}}
\DeclareMathOperator{\tr}{tr}
\DeclareMathOperator{\card}{card}
\newcommand\stab{{\operatorname{STAB}}}
\newcommand{\ignore}[1]{}
\newcommand{\st}[1]{\ket{#1}\bra{#1}}
\newcommand{\otoc}[0]{\text{OTOC}}

\newcommand{\de}[0]{{\operatorname{d}}}

\newcommand{\expf}[1]{\mathrm{exp}\left ( {#1}\right )}

\newcommand{\parent}[1]{\left( {#1} \right)}
\newcommand{\aver}[1]{ \left\langle  {#1}  \right\rangle }

\newcommand{\sym}{\text{sym}}
\newcommand{\pur}{\operatorname{Pur}}

\newcommand{\averu}[1]{\mathbb{E}_{U}\left[#1\right]}

\def\norm#1{\Vert #1\Vert}
\def\CC{{\rm\kern.24em \vrule width.04em height1.46ex depth-.07ex
    \kern-.29em C}}
\def\P{{\rm I\kern-.25em P}}
\def\RR{{\rm
         \vrule width.04em height1.58ex depth-.0ex
         \kern-.04em R}}
\def\bbbone{{\mathchoice {\rm 1\mskip-4mu l} {\rm 1\mskip-4mu l}
{\rm 1\mskip-4.5mu l} {\rm 1\mskip-5mu l}}}
\def\bbbc{{\mathchoice {\setbox0=\hbox{$\displaystyle\rm C$}\hbox{\hbox
to0pt{\kern0.4\wd0\vrule height0.9\ht0\hss}\box0}}
{\setbox0=\hbox{$\textstyle\rm C$}\hbox{\hbox
to0pt{\kern0.4\wd0\vrule height0.9\ht0\hss}\box0}}
{\setbox0=\hbox{$\scriptstyle\rm C$}\hbox{\hbox
to0pt{\kern0.4\wd0\vrule height0.9\ht0\hss}\box0}}
{\setbox0=\hbox{$\scriptscriptstyle\rm C$}\hbox{\hbox
to0pt{\kern0.4\wd0\vrule height0.9\ht0\hss}\box0}}}}
\def\bbbz{{\mathchoice {\hbox{$\sf\textstyle Z\kern-0.4em Z$}}
{\hbox{$\sf\textstyle Z\kern-0.4em Z$}}
{\hbox{$\sf\scriptstyle Z\kern-0.3em Z$}}
{\hbox{$\sf\scriptscriptstyle Z\kern-0.2em Z$}}}}

\usepackage{hyperref}

\begin{document}
\setcounter{secnumdepth}{3}

 \title{Stabilizer R\'enyi Entropy}
\author{Lorenzo Leone}\email{Lorenzo.Leone001@umb.edu}
\affiliation{Physics Department,  University of Massachusetts Boston,  02125, USA}
\author{Salvatore F.E. Oliviero}
\affiliation{Physics Department,  University of Massachusetts Boston,  02125, USA}
\author{Alioscia Hamma}
\affiliation{Physics Department,  University of Massachusetts Boston,  02125, USA}
\affiliation{Université Grenoble Alpes, CNRS, LPMMC, 38000 Grenoble, France}
\begin{abstract}
We introduce a novel measure for the quantum property of nonstabilizerness - commonly known as ``magic'' - by considering the R\'enyi entropy of the probability distribution associated to a pure quantum state given by the square of the expectation value of Pauli strings in that state. We show that this is a good measure of nonstabilizerness from the point of view of resource theory and show bounds with other known measures. The stabilizer R\'enyi entropy has the advantage of being easily computable because it does not need a minimization procedure. We present a protocol for an experimental measurement by randomized measurements. We show that the nonstabilizerness is intimately connected to out-of-time-order correlation functions and that maximal levels of nonstabilizerness are necessary for quantum chaos.
\end{abstract}
 \maketitle

{\em Introduction}.--- Quantum physics is inherently different from classical physics and this difference comes in two layers. First, quantum correlations are stronger than classical correlations and do violate Bell's inequalities\cite{bell64epr,bell1966hidden}. Classical physics can only violate Bell's inequalities at the expense of locality. Second, based on the assumption that $P\ne NP$, quantum physics is exponentially harder to simulate than classical physics\cite{Shor1997Polynomial}. The theory of quantum computation is based on the fact that, by harnessing this complexity, quantum computers would be exponentially faster at solving certain computational tasks\cite{Shor1997Polynomial,grover1996quantum, Childs2010quantum, Lloyd1996simul, Childs2018towards}. 

It is a striking fact that these two layers have a hierarchy: entanglement can be created by means of quantum circuits that can be efficiently simulated on a classical computer\cite{Gottesman:1998hu}. These states are called stabilizer states ($\stab$) and they constitute the orbit of the Clifford group, that is, the normalizer of the Pauli group. Therefore, starting from states in the computational basis, quantum circuits with gates from the Clifford group can be simulated on a classical computer, in spite of being capable of making highly entangled states. The second layer of quantumness thus needs non-Clifford gates. These resources are necessary to unlock quantum advantage. Since there is never a free lunch, non-Clifford resources are harder to implement both at the experimental level and for the sake of error correction\cite{Shor1995dec, Shor1996error, Bennett1996error, Knill1996codes,1997PhDT.......232G}.  Understanding nonstabilizerness in quantum states is of fundamental importance to understand the achievable quantum advantage in schemes of quantum computing\cite{shor1996fault,gottesmann1998fault,Kitaev2003fault,Campbell2017fault} or other quantum information protocols\cite{harrow2004family,Nielsen}. Resource theory of nonstabilizerness has recently found copious applications in magic state distillation and non-Clifford gate synthesis\cite{Campbell2011Catalysis,howard2017application,beverland2019lower,seddon2021quantifying}, as well as classical simulators of quantum computing architectures\cite{bravyi2016improved,bravyi2016resources,bravyi2019simulation,seddon2021quantifying}.

In a broader context, one would like to know what is the bearing of this second layer of quantumness on other fields of physics: from black holes and quantum chaos\cite{leone2021quantum,liu2020many} to quantum many-body theory\cite{liu2020many}, entanglement theory\cite{zhou2020single}, and quantum thermodynamics\cite{YungerHalpern2017}. 

Standard measures of nonstabilizerness are based on general resource theory considerations. A good measure must be stable under operations that send stabilizer states into stabilizer states and faithful, that is,   stabilizer states (and only those) must return zero. Known measures of nonstabilizerness either involve computing an extreme over all the possible stabilizer decompositions of a state and are therefore very hard to compute or cannot anyway be seen as expectation values of an observable\cite{howard2017application,beverland2019lower,liu2020many}.

In this Letter, we define a measure of nonstabilizerness  as the R\'enyi entropy associated to the probability of a state being represented by a given Pauli string. Computing this quantity does not involve a minimization procedure. We also present a protocol for its experimental measurement based on randomized measurements\cite{Enk2012measure,Elben2018Entropies,Brydges2019randomized,Elben2019probing,Elben2020randomized,Zhou2020Estimation}.
We show that, in the context of state synthesis, $O(n)$ magic states are necessary to prepare a Haar-random state. Then we proceed to investigate how much stabilizer entropy a unitary operator can achieve on average on the stabilizer states, that is, the free resources and finally we show that the nonstabilizing power of a quantum evolution can be cast in terms of out-of-time-order correlation functions (OTOCs) and that is thus a necessary ingredient of quantum chaos.

{\em Stabilizer R\'enyi Entropy}.---
In this section, we define a family of nonstabilizerness measures for pure states. Let $\tilde{\mathcal{P}}_{n}$ be the group of all $n$-qubit Pauli strings with phases $\pm 1$ and $\pm i$; then let $\mathcal{P}_n:=\tilde{\mathcal{P}}_{n}/ \aver{\pm i \bbbone}$, the quotient group containing all $+1$ phases and define  $\Xi_{P}(\ket{\psi}):= d^{-1}\bra{\psi}P\ket{\psi}^2$ as the squared (normalized) expectation value of $P$ in the pure state $|\psi\rangle$, with $d\equiv 2^{n}$ the dimension of the Hilbert space of $n$ qubits. Note that
$\sum_{P\in\mathcal{P}_n}\Xi_{P}(\ket{\psi})=\tr\st\psi^2=1$. 
Thus, since $\Xi_{P}(\ket{\psi})\ge 0$ and sum to one,  $\{\Xi_{P}({\ket{\psi}})\}$ is a probability distribution. We can see $\Xi_{P}(\ket{\psi})$ as the probability of finding $P$ in the representation of the state $|\psi\rangle$. We can now define the  $\alpha$-R\'enyi entropies associated to this probability distribution as
\be
M_{\alpha}(\ket{\psi}):=(1-\alpha)^{-1}\log \sum_{P\in\mathcal{P}_n}\Xi_{P}^{\alpha}(\ket{\psi})-\log d,
\ee
where we have introduced a shift of $-\log d$ for convenience. Now let $\Xi(\ket{\psi})$, the vector with $d^2$ entries labeled by $\Xi_{P}(\ket{\psi})$; then we can rewrite the stabilizer $\alpha$-R\'enyi entropy in terms of its $l_{\alpha}$ norm as:
\be
M_{\alpha}(\ket{\psi})=\alpha(1-\alpha)^{-1}\log \norm{\Xi(\ket{\psi})}_{\alpha}-\log d
\ee
The stabilizer R\'enyi entropy is a good measure from the point of view of resource theory. Indeed, it has the following properties:
(i) faithfulness: $M_{\alpha}(\ket{\psi})=0$ iff $\ket{\psi}\in\stab$, otherwise $M_{\alpha}(\ket{\psi})>0$.
    (ii) stability under free operations $C\in\mathcal{C}(\mathcal{H})$: $M_{\alpha}(C\ket{\psi})=M_{\alpha}(\ket{\psi})$; and
   (iii) additivity: $M_{\alpha}(\ket{\psi}\otimes \ket{\phi})= M_{\alpha}(\ket{\psi})+M_{\alpha}(\ket{\phi})$.
The proof can be found in\cite{supp1}.
We are particularly interested in the case  $\alpha=2$:
\be
M_{2}(\ket{\psi})=-\log d\norm{\Xi(\ket{\psi})}_{2}^{2}.
\ee
The above quantity can be rewritten in terms of the projector $Q:=d^{-2}\sum_{P\in\mathcal{P}_n}P^{\otimes 4}$ as:
$M_{2}(\ket{\psi})=-\log d\tr(Q\st{\psi}^{\otimes 4})$. The stabilizer $\alpha$-R\'enyi entropies are upper bounded as
$M_{\alpha}(\ket{\psi})\le \log d$. The proof is elementary:
from the hierarchy of R\'enyi entropies we have that for any $\alpha >0$, $M_{\alpha}(\ket{\psi})\le S_{0}(\ket{\psi})\equiv \log \card({\ket{\psi}})/d$ and then note that $\card({\ket{\psi}})\le d^{2}$, where $\card(\ket{\psi})$ is the number of nonzero entries of $\Xi(\ket{\psi})$. This bound is generally quite loose for pure states. For the stabilizer $2$-R\'enyi entropy we
can obtain a tighter bound: $M_{2}(\ket{\psi})< \log(d+1) -\log 2 $. This is easy to see by picking an Hermitian operator $\rho$ and setting  $\Xi_{\bbbone}({\rho}):=\tr(\rho)=d^{-1}$ and $\Xi_{P}({\rho}):=\tr(P\rho)=d^{-1}(d+1)^{-1}$ for all $P\neq \bbbone$ which maximizes the $2$-R\'enyi entropy by keeping  $\tr\rho=1$ and $\tr\rho^2=1$, although $\rho$ results being nonpositive in general\cite{zhu2016clifford}.

Another useful measure of nonstabilizerness is given by the stabilizer linear entropy, defined as
\be
M_{\text{lin}}(\ket{\psi}):=1-d\norm{\Xi(\ket{\psi})}_{2}^{2}
\ee
which obeys the following properties:
(i) faithfulness: $M_{\text{lin}}(\ket{\psi})=0$ iff $\ket{\psi}\in\stab$, otherwise $M_{\text{lin}}(\ket{\psi})>0$.
   (ii) stability under free operations $C\in\mathcal{C}(\mathcal{H})$: $M_{\text{lin}}(C\ket{\psi})=M_{\text{lin}}(\ket{\psi})$; and
   (iii) upper bound: $M(\ket{\psi})< 1-2(d+1)^{-1}$. 
The proofs are easy consequences of the previous considerations. 

Let us now show how this measure compares to other measures:   the stabilizer nullity\cite{beverland2019lower,jiang2021lower} is defined as $  \nu(\ket{\psi}):=\log d-\log |St(\ket{\psi})|$, where $St(\ket{\psi}):=\{P\in\mathcal{P}_n\,|\, P\ket{\psi}=\pm\ket{\psi}\}$. 

{\bf Proposition:} The stabilizer $\alpha$ R\'enyi entropies are upper bounded by the stabilizer nullity:
\be
 M_{\alpha}(\ket{\psi})\le \nu(\ket{\psi}).\label{alphaR\'enyibound}
\ee
The proof can be found in\cite{supp1}. Notice that for $\alpha=1/2$, the R\'enyi entropy reduces to $M_{1/2}(\ket{\psi})=2\log\mathcal{D}(\ket{\psi})
$, where $\mathcal{D}(\ket{\psi}):=d^{-1}\sum_{P\in\mathcal{P}_n}|\tr(P\st{\psi})|$ is the stabilizer norm defined in\cite{Campbell2011Catalysis}.
More generally, the $\alpha$-R\'enyi entropies (with $\alpha\ge{1}/{2}$) can be upper bounded by twice the logfree robustness of magic\cite{howard2017application} $\mathcal{R}(\ket{\psi}):=\min_{x}\{\norm{x}_1\,|\, \st{\psi}=\sum_{i}x_{i}\sigma_i,\,\sigma_i\in \stab\}$: $M_{\alpha}(\ket{\psi})\le 2\log \mathcal{R}(\ket{\psi})$. The proof of this inequality follows straightforwardly from the hierarchy of R\'enyi entropies and from the bound proven in\cite{howard2017application}: $\mathcal{D}(\ket{\psi})\le \mathcal{R}(\ket{\psi})$ for any state $\ket{\psi}$.

{\bf Example:} In order to understand the advantages of the stabilizer R\'enyi entropy in terms of its computability, let us now compute it for $n$ copies of the magic state $\ket{H}=\frac{1}{\sqrt{2}}(\ket{0}+e^{i\pi/4}\ket{1})$. A straightforward calculation, see\cite{supp1}, yields $M_{\alpha}(\ket{H}^{\otimes n})=(1-\alpha)^{-1}(n\log(2^{1-\alpha}+1)- n)$. 

{\em State synthesis.---}
One of the most useful applications of the resource theory of nonstabilizerness is state synthesis\cite{Campbell2011Catalysis,bravyi2016resources,Campbell2017fault,howard2017application,beverland2019lower}. The main idea is that, given a measure $M$ of nonstabilizerness and two quantum states $\ket{A}$ and $\ket{B}$, if $M(\ket{A})<M(\ket{B})$ one cannot synthesize $\ket{B}$ starting from $\ket{A}$ using stabilizer operations. In this context, we use the stabilizer $2$-R\'enyi entropy to obtain a lower bound on a synthesis of a Haar-random state.

{\bf Theorem (Informal):}  With overwhelming probability, $O(n)$ copies of the magic state $\ket{H}$ are necessary to synthesize a $n$-qubit Haar-random state.

The formal statement and the formal proof can be found in \cite{supp1}.
 
{\em Measuring stabilizer R\'enyi entropy.---} An important feature of the stabilizer $2$-R\'enyi entropy is that it is amenable to be measured in an experiment. As the purity can be measured  via a randomized measurements protocol\cite{Brydges2019randomized,Elben2019probing,Zhou2020Estimation}, we show that a suitable randomized measurements of Clifford operators can return $M_2$. Let $\ket{\psi}$ be the quantum pure state. Randomly choose an operator $C\in\mathcal{C}(2^n)$ and operate it on the state $C\ket{\psi}$, then measure $C\ket{\psi}$ in the computational basis $\{\ket{\mathbf{s}}\}\equiv \{s=0,1\}^{\otimes n}$. For a given $C$, by repeated measurements one can estimate the probability $P(s|C):=|\langle s|C|\psi\rangle|^2$. Define the vector of four $n$-bit strings $\vec{\mathbf{s}}=(\mathbf{s}_1,\mathbf{s}_2,\mathbf{s}_3,\mathbf{s}_4)$ and denote the binary sum of these strings as $\|\vec{\mathbf{s}}\|\equiv \mathbf{s}_1\oplus\mathbf{s}_2\oplus\mathbf{s}_3\oplus\mathbf{s}_4 $. Then the stabilizer  $2$-R\'enyi entropy is equal to
\be
M_{2}(\ket{\psi})=-\log \sum_{\vec{\mathbf{s}}}(-2)^{-\|\vec{\mathbf{s}}\|}\mathcal{Q}(\vec{\mathbf{s}})-\log d
\ee
where $\mathcal{Q}(\vec{\mathbf{s}}):=\mathbb{E}_{C}P(\mathbf{s}_1|C)P(\mathbf{s}_2|C)P(\mathbf{s}_3|C)P(\mathbf{s}_4|C)$ is the expectation value over the randomized measurements of the Clifford operator $C$. For a proof, see\cite{supp1}.

{\em Extension to mixed states.---} The stabilizer R\'enyi entropy can be extended to mixed states. We define the free resources as the states of the form $\chi=d^{-1}(\bbbone+\sum_{P\in G}\phi_{P}P)$ with $G\subset \mathcal{P}_{n}$ a subset of the Pauli group with $0\le |G|\le d-1$. 
Then, we define the stabilizer $2$-R\'enyi entropy of the mixed state $\rho$ as
\be\label{mixedM}
\tilde{M}_{2}(\rho):=M_{2}(\rho)-S_{2}(\rho)
\ee
with $S_2(\rho)$ being the $2$-R\'enyi entropy of $\rho$ and $M_{2}(\rho):=-\log d\tr(Q\rho^{\otimes 4})$. This quantity is again faithful as it is zero only on the free resources, is invariant under Clifford operations 
 $C\in\mathcal{C}(d)$ then $\tilde{M}_{2}(C\rho C^{\dag})=\tilde{M}_{2}(\rho )$ and has additivity: $\tilde{M}_{2}(\rho\otimes\sigma)=\tilde{M}_{2}(\rho)+\tilde{M}_{2}(\sigma)$. As a corollary, if $\chi$ is a stabilizer state then $\tilde{M}_{2}(\rho\otimes \chi)=\tilde{M}_{2}(\rho)$.  The proof is to be found in\cite{supp1}. Numerical evidence also suggests that $\tilde{M}_{2}$ is nonincreasing under partial trace. The same randomized protocol can also be employed to measure $\tilde{M}_{2}(\rho)$.

{\em Nonstabilizing power}.--- In this section, we want to address the problem of how much nonstabilizerness can be produced by a unitary operator, e.g. a quantum circuit. We therefore restrict our attention to pure states. We define the nonstabilizing power of a unitary operator $U$ as 
\be\label{magicpowerdef}
\mathcal{M}(U):=\frac{1}{|\stab|}\sum_{\ket{\psi}\in \stab}M(U\ket{\psi})
\ee
where $M(\ket{\psi})$ is one of the entropic measures introduced in the previous section, i.e. one of the stabilizer $\alpha-$R\'enyi entropy $M_{\alpha}(\ket{\psi})$ or the stabilizer linear entropy $M_{\text{lin}}(\ket{\psi})$. Also the nonstabilizing power is $(i)$ invariant under free operations, that is, $ \mathcal{M}(U)=\mathcal{M}(C_{1}U)=\mathcal{M}(UC_{2})=\mathcal{M}(C_{1}UC_{2})$, with $C_{1},C_{2}\in\mathcal{C}(d)$ and $(ii)$ is faithful, that is, $ \mathcal{M}(U)=0$ for  the free operations $U\in\mathcal{C}(d)$ and is greater than zero otherwise. A proof of these properties is in\cite{supp1}.

The relationship between the $2$-R\'enyi nonstabilizing power and the linear nonstabilizing power follows easily from the Jensen inequality:
\be\label{lin2R\'enyibound}
\mathcal{M}_2(U)\ge -\log(1-\mathcal{M}_{\text{lin}}(U))
\ee
The linear nonstabilizing power can be computed explicitly by averaging the fourth tensor power of the Clifford group:
$\mathcal{M}_{\text{lin}}(U)=1-4(4+d)^{-1}-d(4+d)^{-1}
D_{+}^{-1}\tr(U^{\otimes 4}QU^{\dag\otimes4}\Pi_{\sym})$, with $\Pi_{\sym}:=\frac{1}{4!}\sum_{\pi\in \mathcal{S}_{4}}T_{\pi}$ the projector onto the completely symmetric subspace of the permutation group $\mathcal{S}_{4}$, $Q=d^{-2}\sum_{P}P^{\otimes 4}$ and $D_{+}\equiv\tr{(Q\Pi_{\sym})}=(d+1)(d+2)/6$. The proof can be found in\cite{supp1}. This result, through Eq.(\ref{lin2R\'enyibound}), also gives a lower bound to the $2$-R\'enyi nonstabilizing power. In the following, we provide some useful results on the linear nonstabilizing power (and, through lower bounds, for the $2$-R\'enyi nonstabilizing power). First of all, we provide a characterization of those unitaries that have zero power: the linear nonstabilizing power $\mathcal{M}_{\text{lin}}(U)= 0$ if and only if $[Q\Pi_{\sym},U^{\otimes 4}]=0$, see\cite{supp1} for the proof. A second interesting result is a characterization of this quantity in terms of the operator $\Delta Q\Pi_\sym := U^{\dag\otimes 4}Q\Pi_\sym U^{\otimes 4}-Q\Pi_\sym$, that is, the difference between the operator $Q\Pi_\sym $ after and before unitary evolution through $U^{\otimes 4}$. We have
$\mathcal{M}_{\text{lin}}(U)=d2^{-1}D_{+}^{-1}(d+4)^{-1}\norm{\Delta Q\Pi_\sym}_{2}^{2}$ which follows straightforwardly from $\norm{\Delta Q\Pi_\sym}_{2}^{2}=2D_{+}-2\tr(U^{\otimes 4}QU^{\dag\otimes 4}Q\Pi_\sym)$. Then again one can apply the bound Eq.(\ref{lin2R\'enyibound}) in this form.

After having characterized the nonstabilizing power of a unitary $U$, we are interested in knowing what is the average value that this quantity attains over the unitary group $\mathcal{U}(d)$. We obtain:
\be
\mathbb{E}_{U}\left[\mathcal{M}_{\text{lin}}(U)\right]=1-4(d+3)^{-1}\label{avulin}
\ee
and consequently the $2$-R\'enyi nonstabilizing power is lower bounded by $\mathbb{E}_{U}\left[\mathcal{M}_{2}(U)\right]\ge \log(d+3)-\log 4\label{avulog}$. 
The proof can be found in\cite{supp1}. This average is also typical. The linear nonstabilizing power indeed shows strong typicality with respect to $U\in \mathcal{U}(d)$:
\ba
\hspace{-0.4cm}\text{Pr}\parent{\left|\mathcal{M}_{\text{lin}}(U)-\averu{\mathcal{M}_{\text{lin}}(U)}\right|\ge \epsilon}\le 4e^{-Cd\epsilon^{2}}
\ea
where $C=O(1)$. In other words, the overwhelming majority of unitaries attains a nearly maximum value of $\mathcal{M}_{\text{lin}}(U)= 1-\Theta(d^{-1})$. 
For a proof, see\cite{supp1}. As a corollary, the average $2$-R\'enyi nonstabilizing power over the full unitary group $\mathcal{U}(d)$ saturates the bound up to an exponentially small error.  Note that, because of the left and right invariance of the Haar measure over groups, the average stabilizer $2$-R\'enyi entropy over all the set of pure states is equal to the average $2$-R\'enyi nonstabilizing power over the unitary group, namely $\mathbb{E}_{\ket{\psi}}\left[M_{2}(\ket{\psi})\right]=\averu{\mathcal{M}_{2}(U)}$. To conclude this section, let us show how the nonstabilizing power lower bounds the T count $t(U)$, i.e. the minimum number of $T$ gates needed in addition of Clifford resources to obtain a given unitary operator\cite{jiang2021lower}:
\ba\label{ttt}
t(U)\ge -\log_2\left(d-(4+d)\mathcal{M}_{\text{lin}}\right)+\log_2(d+3)-2.\hspace{0.8cm}
\ea
The proof can be found in\cite{supp1}. According to the typicality result, for a generic $U\in\mathcal{U}(d)$, with overwhelming probability, one obtains $t(U)\gtrsim \Theta(n)$.

\ignore{
{\em Examples.---} 
Let us now provide a few applications of nonstabilizing power when $U$ is equal to a $k-$fold tensor power of single $T$-gates and for $U$ being a single Toffoli gate. In the following, the result does not depend on which qubits these gates are applied to. From direct evaluation of Eq.\eqref{boundlinmag}, we find:

{\em Example 1 - $T$-gate}. The linear nonstabilizing power for a number $k$ of $T$-gates is equal to:
\be 
\mathcal{M}_{\text{lin}}(T^{\otimes k}\otimes \bbbone^{\otimes(n-k)})
=1-24\times 2^{-5k}+\Omega(d^{-1})
\ee

{\em Example 2 - Toffoli gate}. The linear nonstabilizing power for a single nonstabilizerness gate reads:
\be 
\mathcal{M}_{\text{lin}}(CCNOT\otimes \bbbone^{\otimes(n-3)})=21/32+\Omega(d^{-1})
\ee 
Notice that the linear nonstabilizing power of three $T-$gates is greater than a single Toffoli gate, both being three-qubit gates.
}

{\em Nonstabilizerness and chaos}.--- Having defined a measure of nonstabilizing power, we now use it to investigate some important questions in many-body quantum physics and quantum chaos theory. In\cite{leone2021quantum}, it was shown that, in order to obtain the typical behavior of the eight-point out-of-time-order correlation functions ($8$-OTOC) for universal unitaries, a number of $T$ gates of order $\Theta(N)$ was both necessary and sufficient. The universal behavior of $8$-OTOC is a mark of the onset of quantum chaos\cite{leone2021quantum}.  Since the $T$ gates are non-Clifford resources, this raises the more general question of what is the amount of nonstabilizerness necessary to drive a quantum system toward quantum chaos. In\cite{leone2021quantum}, the setting is that of a Clifford circuit doped by $k$ layers of non-Clifford one qubit gates, e.g., the $\theta$-phase gates, what we call $k$-doped random quantum Clifford circuit\cite{zhou2020single,gross2020quantum,leone2021quantum,oliviero2021transitions}. We start addressing the question of what is the nonstabilizing power associated to such circuits. We can show the following.

{\bf Proposition:} The nonstabilizing power is monotone under a $k$-doped random quantum circuit and it is given by
\be
\mathbb{E}_{\mathcal{C}_k}\!\left[\mathcal{M}_{\text{lin}}(U)\right]\!=\!1\!-(3+d)^{-1}\! \left(4+\!\!(d-1)f(\theta)^{k}\right)
\ee
with $f(\theta)= \left(\frac{7d^2-3d +d(d + 3)\cos(4\theta)-8}{8(d^2-1)}\right)\le 1$. The proof can be found in\cite{supp1}. Note that iff $k=\Theta(n)$ then $\mathbb{E}_{\mathcal{C}_k}[\mathcal{M}_{\text{lin}}(U)]=\averu{\mathcal{M}_{\text{lin}}(U)}$, unless, of course, $\theta=\pi/2$ in which case the phase gate is in the Clifford group and $f=1$. This proposition shows how nonstabilizerness increases with non-Clifford doping. We notice that nonstabilizerness will converge exponentially fast to the universal maximal value with the number $k$ of non-Clifford gates used. This is the same type of behavior shown by the $8$-OTOCs\cite{leone2021quantum}. 

At this point, we are ready to show a direct connection between the stabilizer R\'enyi entropy and the OTOCs. We have the following:

{\bf Theorem:} The linear nonstabilizing power is equal to the fourth power of the $2$-OTOC of the Pauli operators $P_1,P_2$ averaged over both all the initial states with the Haar measure and over the Pauli group, that is,
\ba
\mathcal{M}_{\text{lin}}(U)&=&1-4(4+d)^{-1}-d^2(d+3)4^{-1}(d+4)^{-1}\times\nonumber\\
&&\mathbb{E}_{\ket{\psi}}\left[\langle\otoc_{2}(\tilde{P}_1,P_{2},\psi)^{4}\rangle_{P_1,P_2}\right]\label{otocpsiM}
\ea
where $\aver{\cdot}_{P_1,P_2}$ is the average over the Pauli group $\mathcal{P}_n$, $\mathbb{E}_{\ket{\psi}}\left[\cdot\right]$ is the Haar average over set of pure states and
$\otoc_{2}(\tilde{P}_1,P_{2},\psi):=\langle\psi|\tilde{P}_1P_2|\psi\rangle$, 
where $\tilde{P}_1\equiv U^{\dag}P_1U$.
The proof can be found in\cite{supp1}. As a corollary, we can bound the $2$-R\'enyi nonstabilizing power through the linear nonstabilizing power. 

As we can see, the average fourth power of the $2$-OTOC is related to the same moment of the Haar distribution of the following averaged  eight-point out-of-time-order correlation function: $\aver{\otoc_{8}}:=\langle d^{-1}\tr(\tilde{P}_1P_{2}P_3P_4\tilde{P}_{1}P_2P_4P_5\tilde{P}_1P_{2}P_5P_6\tilde{P}_1P_{2}P_6P_3)\rangle$, where the average $\aver{\cdot}$ is taken over all the Pauli operators $P_{i}$ for $i=1,\ldots, 6$. One can therefore show that the linear nonstabilizing power is related to the $8$-OTOC as follows

{\bf Theorem:} The linear nonstabilizing power can be expressed as an eight-point OTOC up to an exponentially small error in $d$:
\ba
&&\mathcal{M}_{\text{lin}}(U)\simeq 1-\frac{1}{(d+4)}\left(4+\frac{d^2(d+3)}{4}\aver{\otoc_8}\right)\nonumber
\ea
The proof can be found in\cite{supp1} and it relies on the fact that the $2$-OTOCs have strong typicality with respect to $\ket{\psi}$. We can comment on this last result: in order for the $8$-OTOCs to attain the Haar value, $\sim d^{-4}$ associated to quantum chaotic behavior (cfr.\cite{supp1}), then the nonstabilizing power of $U$ needs to be $M_{\text{lin}} (U) \simeq 1-4/d$ for large dimension $d$. So only unitaries with maximal non stabilizing power (up to an exponentially small error ) can be chaotic.

{\em Conclusions}.--- Harnessing the power of quantum physics to obtain an advantage over classical information processing is at the heart of the efforts to build a quantum computer and finding quantum algorithms. Quantumness beyond classical simulability is quantified in terms of how many non-Clifford resources are necessary (nonstabilizerness), and this notion has been colloquially dubbed ``magic''. This information-theoretic notion is also involved  - beyond quantum computation - in physical processes like thermalization, quantum thermodynamics, black holes dynamics, and the onset of quantum chaotic behavior\cite{Sarkar2020resource,White2021cft,liu2020many,leone2021quantum}. 
In this Letter, we have shown a new measure of nonstabilizerness 
in terms of the R\'enyi entropies of a probability distribution associated to the (squared) expectation values of Pauli strings and show that this is a good measure from the point of view of resource theory. This quantity can be measured experimentally through a randomized measurement protocol. Thanks to this new measure, we can define the notion of nonstabilizing power of a unitary evolution and show that the onset of quantum chaos requires a maximal amount of the stabilizer R\'enyi entropy. 


{\em Acknowledgments.}--- We acknowledge support from NSF award no. 2014000. 
%

\vspace{-2cm}
\end{document}


\setcounter{secnumdepth}{3}
\renewcommand{\theequation}{S.\arabic{equation}}
\title{Supplemental Material: Stabilizer R\'enyi Entropy}
\author{Lorenzo Leone}\email{Lorenzo.Leone001@umb.edu}
\affiliation{Physics Department,  University of Massachusetts Boston,  02125, USA}
\author{Salvatore F.E. Oliviero}
\affiliation{Physics Department,  University of Massachusetts Boston,  02125, USA}
\author{Alioscia Hamma}
\affiliation{Physics Department,  University of Massachusetts Boston,  02125, USA}
\affiliation{Université Grenoble Alpes, CNRS, LPMMC, 38000 Grenoble, France}
\maketitle
\onecolumngrid
\section{Stabilizer $\alpha$-R\'enyi entropy}
\subsection{Properties of the stabilizer R\'enyi entropy}\label{proofpropertiesalphamagic} 
In this section, we give a proof of the properties $(i)$, $(ii)$ and $(iii)$ for the stabilizer $\alpha-$R\'enyi entropy. Let us rewrite $M_{\alpha}(\ket{\psi})$ as
\be
M_{\alpha}(\ket{\psi})=(1-\alpha)^{-1}\log d^{-\alpha} \sum_{P\in\mathcal{P}_n}\aver{\psi|P|\psi}^{2\alpha}-\log d
\ee
\begin{enumerate}[label=(\roman*)]

\item If $\ket{\psi}$ is a stabilizer state $\aver{\psi|P|\psi}=\pm 1$ only for $d$ mutually commuting Pauli operators  and thus $M_{\alpha}(\ket{\psi})=\frac{1}{1-\alpha}\log d^{-\alpha+1}-\log d=0$. Conversely, let $M_{\alpha}(\ket{\psi})=0$ with $\alpha\neq 0,1$, this means that $\norm{\Xi(\ket{\psi})}_{\alpha}=d^{(1-\alpha)/\alpha}$. Recalling that the purity of the state must be one, we write:
\be
\begin{cases}
\sum_{P}\tr^{2\alpha}(P\st{\psi})=d\\
\sum_{P}\tr^{2}(P\st{\psi})=d
\end{cases}
\ee
the only possibility for the two equations to be satisfied at the same time is that $\st{\psi}$ has the following form:
\be
\st{\psi}=\frac{1}{d}\sum_{P}\phi_PP
\label{stabstate}
\ee
where $\phi_{P}$ are $\pm 1$, this is a stabilizer state\cite{zhu2016clifford}. For $\alpha=0$ we have that $\card(\Xi(\ket{\psi}))=d$ and this is the case only for states written in the form of Eq. \eqref{stabstate}. Lastly, the proof for $\alpha=1$  follows by continuity. 
\item  Let $C\in\mathcal{C}_n$ the Clifford group on $n-$qubits. The action of a Clifford operator on a Pauli operator $P\in\mathcal{P}_n$ reads $C^{\dag}PC=\phi Q$ where $Q\in\mathcal{P}_n$ $\phi=\pm 1$, thus trivially the stabilizer $\alpha-$R\'enyi is conserved for free operations. 
\item Consider $\ket{\psi}\otimes \ket{\phi}$ where $\ket{\psi}\in\mathcal{S}(\mathcal{H}_1)$ and $\ket{\phi}\in\mathcal{S}(\mathcal{H}_2)$ with $d_{1}=\dim\mathcal{H}_1$ and $d_{2}=\dim\mathcal{H}_2$, then:
\be
M_{\alpha}(\ket{\psi}\otimes \ket{\phi})=\frac{1}{1-\alpha}\log d_{1}^{-\alpha}d_{2}^{-\alpha} \sum_{P\in\mathcal{P}_n}\aver{\psi\otimes \phi|P|\psi\otimes \phi}^{2\alpha} -\log d_1-\log d_2
\ee
For every $P\in\mathcal{P}_n$ we can write it as $P=P_{1}\otimes P_{2}$, thus $\sum_{P\in\mathcal{P}_n}\aver{\psi\otimes \phi|P|\psi\otimes \phi}^{2\alpha}=\sum_{P_{1}\in\mathcal{P}_{n}^{1},P_{2}\in\mathcal{P}_{n}^{2}}\aver{\psi|P_1|\psi}^{2\alpha}\aver{\phi|P_2|\phi}^{2\alpha}$ and thus:
\ba
M_{\alpha}(\ket{\psi}\otimes \ket{\phi})&=&\frac{1}{1-\alpha}\log d_{1}^{-\alpha} \sum_{P_1\in\mathcal{P}_{n}^{1}}\aver{\psi|P_1|\psi}^{2\alpha}+\frac{1}{1-\alpha}\log d_{2}^{-\alpha} \sum_{P_2\in\mathcal{P}_{n}^{2}}\aver{\phi|P_2|\phi}^{2\alpha}\nonumber\\ &-&\log d_1-\log d_2=M_{\alpha}(\ket{\psi})+M_{\alpha}( \ket{\phi})
\ea
this concludes the proof.
\end{enumerate}

\subsection{Bounding the stabilizer R\'enyi entropy}\label{proofboundrenyientropy}
Int this section, we provide a proof for the bound to the stabilizer R\'enyi entropy given in Eq.(5). We first prove that:
\be
|St(\ket{\psi})|\le \frac{d^2}{\card(\ket{\psi})}
\ee
where we recall that $St(\ket{\psi}):=\{P\in\mathcal{P}_n\,|\,P\ket{\psi}=\pm\ket{\psi}\}$. Consider a pure state $\ket{\psi}$ and its expansion in the Pauli basis, i.e. $\st{\psi}=d^{-1}\sum_{P}\aver{\psi|P|\psi}P$. The number of elements of this sum is, by definition, equal to $\card(\ket{\psi})$. Let $|St(\ket{\psi})|=S$. First, it's easy to see that $S$ must be equal to $2^s$ for some $s$, because if $P_1,\dots, P_s$ stabilize $\ket{\psi}$ also the group generated by $P_1,\dots,P_2$ stabilizes $\ket{\psi}$ and $|\aver{P_1,\dots, P_s}|=2^s$. Then, in order to $P_1,\dots,P_s$ to stabilize $\ket{\psi}$ they must commute among each other and commute with all the $\card(\ket\psi)$ Pauli operators in the Pauli decomposition of $\ket{\psi}$. Now, since $s$ mutually commuting Pauli operators commute at most with $d^2/2^s$ Pauli operators\cite{sarkar2019sets}, $\card(\ket{\psi})$ cannot exceed $d^2/2^s$ otherwise $\aver{P_1,\dots, P_s}$ can never stabilize $\ket{\psi}$. We conclude that $\card(\ket{\psi})\le d^2/2^s\equiv d^2/|St(\ket{\psi})|$. In order to conclude the proof is sufficient to apply the log to this inequality and get 
\be
M_0(\ket{\psi})\equiv \log\frac{\card(\ket{\psi})}{d}\le \nu(\ket\psi),
\ee
then from the hierarchy of R\'enyi entropies one easily gets the upper bound of Eq.(5).

\subsection{An example: stabilizer $\alpha-$Renyi entropy for the magic state}\label{magicstate}
The magic state is a single qubit quantum state defined as $\ket{H}:=\frac{1}{\sqrt{2}}(\ket{0}+e^{i\pi/4}\ket{1})$. We can write the density matrix $\st{H}$ in the Pauli decomposition as:
\be
\st{H}=\frac{1}{2}\bbbone+\frac{1}{2\sqrt{2}}(X-Y)
\ee
In order to compute the stabilizer $\alpha-$Renyi entropy of the $n-$fold tensor power, we need to find the vector $\Xi(\ket{H}^{\otimes n})$ containing the probability distribution in the Pauli basis. Thus we need the following expansion:
\be
\st{H}^{\otimes n}=\frac{1}{2^n}\left(\bbbone+\frac{1}{\sqrt{2}}(X-Y)\right)^{\otimes n}
\ee
It can be easily checked that among the $4^{n}$ possible Pauli operators we have only $3^n$ non zero entries for $\Xi(\ket{H}^{\otimes n})$ and $\Xi_{P}(\ket{H}^{\otimes n})=\frac{1}{2^{n+k}}$ for $2^{k}\binom{n}{k}$ Pauli operators for $k=0,\dots,n$. Indeed, note that:
\ba
\sum_{k=0}^{n}\binom{n}{k}2^{k}&=&3^n\\
\sum_{k=0}^{n}\binom{n}{k}2^{k}\frac{1}{2^{k+n}}&=&1\\
\ea
i.e. the non zero entries are $3^n$ and it is a normalized probability distribution. The stabilizer $\alpha-$Renyi entropy:
\be
M_{\alpha}(\ket{H}^{\otimes n})=\frac{1}{1-\alpha}\log \sum_{k=0}^{n}\binom{n}{k}2^{k}2^{-\alpha(n+k)}-n=\frac{1}{1-\alpha}(n\log(2^{1-\alpha}+1)- n)
\label{magicstaterenyientropy}
\ee


\section{Randomized measurements and $2-$Stabilizer R\'enyi entropy: proof of Eq.$(6)$}
In this section we are going to provide the proof of Eq. $(6)$, i.e. we prove that the stabilizer $2-$R\'enyi entropy can be measured via a randomized measurements protocol\cite{Brydges2019randomized}. Let $\ket{\psi}$ be the quantum \textit{pure} state we want to measure the value of the stabilizer $2-$R\'enyi entropy of. First, randomly choose a unitary operator $C\in\mathcal{C}(2^n)$ and operate it on the state $C\ket{\psi}$, then measure $C\ket{\psi}$ in the computational basis $\{\ket{\mathbf{s}}\}\equiv \{s=0,1\}^{\otimes n}$. We denote by $\mathbf{s}=(s^1,\ldots,s^n)$ the n-bit string corresponding to the computational basis state. For a given $C$, by repeated measurements one can estimate the probability $P(s|C):=|\langle s|C|\psi\rangle|^2$. First, define:
\be
\mathcal{Q}(\mathbf{s}_{1},\mathbf{s}_2,\mathbf{s}_3,\mathbf{s}_4):=\mathbb{E}_{C}P(\mathbf{s}_1|C)P(\mathbf{s}_2|C)P(\mathbf{s}_3|C)P(\mathbf{s}_4|C)
\ee
i.e. the average probability over $C$ of all the computational basis vectors $\ket{\mathbf{s}_1},\ket{\mathbf{s}_2},\ket{\mathbf{s}_3}$ and $\ket{\mathbf{s}_4}$.  Define now $\ket{\vec{\mathbf{s}}}:=\ket{\mathbf{s}_1\mathbf{s}_2\mathbf{s}_3\mathbf{s}_4}$. One can express the above as:
\be
\mathcal{Q}(\vec{\mathbf{s}})=\mathbb{E}_{C}\tr(\ket{\vec{\mathbf{s}}}\bra{\vec{\mathbf{s}}}C^{\otimes 4}\st{\psi}^{\otimes 4}C^{\dag\otimes 4})
\label{definitionQ}
\ee

Now, define the diagonal operator $\hat{O}:=\sum_{\vec{\mathbf{s}}}O(\vec{\mathbf{s}})\ket{\vec{\mathbf{s}}}\bra{\vec{\mathbf{s}}}$.
We then make the ansatz:
\ba
\mathbb{E}_{C} \tr \left[ \hat{O} C^{\otimes 4}\st{\psi}^{\otimes 4}C^{\dag\otimes 4}\right]
=\tr(Q\st{\psi}^{\otimes 4})
\ea
that is,
\be
\sum_{\vec{\mathbf{s}}}O(\vec{\mathbf{s}})\mathcal{Q}(\vec{\mathbf{s}})=\tr(Q\st{\psi}^{\otimes 4})
\label{ansatz}
\ee
By permuting in the trace we have
\be
\mathbb{E}_{C}\tr(C^{\otimes 4}\hat{O}C^{\otimes 4}\st{\psi}^{\otimes 4})=\tr(Q\st{\psi}^{\otimes 4})
\ee
for the ansatz in Eq. \eqref{ansatz} to be satisfied, we need to require that:
\be
\mathbb{E}_{C}[C^{\otimes 4}\hat{O}C^{\dag\otimes 4}]=Q
\ee
By the Weingarten formula for the Clifford group, we have \cite{leone2021quantum}:
\be
\mathbb{E}_{C}[C^{\otimes 4}\hat{O}C^{\dag\otimes 4}]=\sum_{\pi\sigma}W^{+}_{\pi\sigma}\tr(\hat{O}QT_{\pi})QT_{\sigma}+W^{-}_{\pi\sigma}\tr(\hat{O}Q^{\perp}T_{\pi})Q^{\perp}T_{\sigma}
\ee
Since $\mathbb{E}_{C}[C^{\otimes 4}QC^{\dag\otimes 4}]=Q$ and since the twirling is a map \textit{many to one}, a direct application of the above equation is to require the following equations to be satisfied:
\be
\begin{cases}
\tr(\hat{O}QT_{\pi})=\tr(QT_{\pi}), \quad \forall \pi\in \mathcal{S}_{4}\\
\tr(\hat{O}T_{\sigma})=\tr(QT_{\sigma}), \quad \forall \sigma\in \mathcal{S}_{4}
\end{cases}
\ee
In order to solve the above, we note that we can write $\tr(\hat{O}QT_{\pi})=\tr(\hat{o}Q_{2}T_{\pi}^{(2)})^{n}$, where $Q_{2}=\bbbone^{\otimes 4}+X^{\otimes 4}+Y^{\otimes 4}+Z^{\otimes 4}$, $T_{\sigma}^{(2)}=\sum_{s_{1},s_{2},s_{3},s_{4}}\ket{\pi(s_{1}s_{2}s_{3}s_{4})}\bra{s_{1}s_{2}s_{3}s_{4}}$ and $s_{i}=\{0,1\}$ for $i=1,2,3,4$, i.e. $T_{\pi}^{(2)}$ is the permutation operator corresponding to the permutation $\pi\in \mathcal{S}_{4}$ defined on four copies of $\mathbb{C}^{2}$, and finally $\hat{o}^{\otimes n}=S\hat{O}S^{-1}$, where $S\ket{s_{1}\dots s_{n}}^{\otimes 4}=\ket{s_1}^{\otimes 4}\otimes \dots\otimes \ket{s_n}^{\otimes 4}$ and $\hat{o}$ is a diagonal operator defined on four copies of $\mathbb{C}^{2}$:
\be
\hat{o}:=\sum_{\vec{s}}o(\vec{s})\ket{\vec{s}}\bra{\vec{s}}
\ee
where $\vec{s}:=s_{1}s_{2}s_{3}s_{4}$. The above operations are legitimate because $(\mathbb{C}^{2\otimes n})^{\otimes 4}\simeq (\mathbb{C}^{2\otimes 4})^{\otimes n}$. Thus, we can solve the following equations:
\be
\begin{cases}
\tr(\hat{o}Q_{2}T_{\pi}^{(2)})=\tr(Q_{2}T_{\pi}^{(2)}), \quad \forall \pi\in\mathcal{S}_{4}\\
\tr(\hat{o}T_{\sigma}^{(2)})=\tr(Q_{2}T_{\sigma}^{(2)}), \quad \forall \sigma\in\mathcal{S}_{4}
\end{cases}
\ee
Solving the above relations we find that (one of the solution):
\be
o(\vec{s})\equiv o(s_{1},s_{2},s_{3},s_{4})=(-2)^{-s_{1}\oplus s_{2}\oplus s_{3}\oplus s_{4}}
\ee
where $\oplus$ is the binary sum between bits of $0$s and $1$s. By applying the operator $S$ defined above we find:
\be
O(\vec{\mathbf{s}})\equiv O(\mathbf{s}_{1},\mathbf{s}_{2},\mathbf{s}_{3},\mathbf{s}_{4})=(-2)^{-\mathbf{s}_1\oplus\mathbf{s}_2\oplus\mathbf{s}_3\oplus\mathbf{s}_4}
\label{weigths}
\ee
where we defined $\mathbf{s}_1\oplus\mathbf{s}_2\oplus\mathbf{s}_3\oplus\mathbf{s}_4:=\sum_{i=1}^{n}s_{1}^{i}\oplus s_{2}^{i}\oplus s_{3}^{i}\oplus s_{4}^{i}$. Thus, by combining the average measurements $\mathcal{Q}(\mathbf{s}_1,\mathbf{s}_2,\mathbf{s}_3,\mathbf{s}_4)$ in Eq. \eqref{definitionQ} with weights given by Eq. \eqref{weigths} we obtain the stabilizer $2-$R\'enyi entropy:
\be
M_{2}(\ket{\psi})=-\log \sum_{\vec{\mathbf{s}}}(-2)^{-\norm{\vec{\mathbf{s}}}}\mathcal{Q}(\vec{\mathbf{s}})-\log d
\ee
where $\norm{\vec{\mathbf{s}}}:=\mathbf{s}_1\oplus\mathbf{s}_2\oplus\mathbf{s}_3\oplus\mathbf{s}_4$.

\section{Stabilizer $2-$R\'enyi entropy for mixed state}
The aim of this section is to extend the definition of stabilizer $2-$R\'enyi entropy to mixed states. First of all, let us define the set of free state for our non-stabilizerness measure. We define $S$ to be a stabilizer state iff $\chi$ is positive semidefinite and its decomposition in the Pauli basis reads:
\be
\chi=\frac{\bbbone}{d}+\frac{1}{d}\sum_{P\in G}\phi_{P}P
\label{stabilizermixed}
\ee
where $G$ is a subset of the $n-$qubit Pauli group with cardinality $0<|G|<d-1$ and $\phi_{P}\in\{-1,+1\}$. Note that $\chi$, defined in Eq. \eqref{stabilizermixed}, is hermitian and has trace one. Moreover, we have that $\chi$ is a pure stabilizer state iff $|G|=d-1$\cite{zhu2016clifford}. In order to extend the stabilizer $2-$R\'enyi entropy to mixed states, recall the peculiar property of the Pauli decomposition for pure states: for a ket $\ket{\psi}$, one has $\sum_{P}\Xi_{P}(\ket{\psi})=1$, where $\Xi_{P}=d^{-1}\tr(P\st{\psi})^{2}$ and thus $\{\Xi_{P}(\ket{\psi}\}$ is a probability distribution. This property fails for a generic mixed state $\rho$, indeed:
\be
d^{-1}\sum_{P}\tr(P\rho)^{2}=\tr(\rho^2)\equiv\pur(\rho)
\ee
and thus $\{\Xi_{P}(\rho)\}$ is not a normalized, hence is not a probability distribution. Thus, we lay down the extended definition of stabilizer $2-$R\'enyi entropy for mixed states $\rho$ as:
\be
\tilde{M}_{2}(\rho):=M_{2}(\rho)-S_{2}(\rho)
\ee
where $S_{2}(\rho):=-\log\pur(\rho)$ is just the $2-$R\'enyi entanglement entropy and $M_{2}(\rho)\equiv -\log d\tr(Q\rho^{\otimes 4})$. First, we can write the above in the following nice form:
\be
\tilde{M}_{2}(\rho)\equiv-\log\frac{\sum_{P}\tr^{4}(P\rho)}{\sum_{P}\tr^{2}(P\rho)}
\ee
where the sum is over the whole Pauli group of $n-$qubit $\mathcal{P}_{n}$. It obeys to the following properties:
\begin{theorem*}
The stabilizer $2-$R\'enyi entropy $\tilde{M}_{2}(\rho)$, defined on mixed states $\rho$, obeys to the following properties:
\begin{enumerate}[label=(\roman*)]
    \item faithfulness: $\tilde{M}_{2}(\rho)=0$ iff $\rho$ is a stabilizer state of the form of Eq. \eqref{stabilizermixed}.
    \item invariance under Clifford operations: $C\in\mathcal{C}(d)$ then $\tilde{M}_{2}(C\rho C^{\dag})=\tilde{M}_{2}(\rho )$.
    \item additivity: $\tilde{M}_{2}(\rho\otimes\sigma)=\tilde{M}_{2}(\rho)+\tilde{M}_{2}(\sigma)$. As a corollary, if $S$ is a stabilizer state then $\tilde{M}_{2}(\rho\otimes S)=\tilde{M}_{2}(\rho)$.
\end{enumerate}
\end{theorem*}
\begin{proof}

\begin{enumerate}[label=(\roman*)]
    \item Start with $\chi$ being a stabilizer state of the form in Eq. \eqref{stabilizermixed}. For all $P\in G$ we have $\tr^{2}(P\chi)=1$ and for any $P\not\in G$ we have $\tr(P\chi)=0$, so:
\be
\sum_{P}\tr^{2}(P\chi)=\sum_{P}\tr^{4}(P\chi)=|G|+1
\ee
thus their ratio is one and consequently $\tilde{M}_{2}(\chi)=0$. Conversely, if $\tilde{M}_{2}(\rho)=0$ for $\rho$ being a positive semi-definite, hermitian and unit trace operator, we have the following relation to be satisfied:
\be
\sum_{P}\tr^{2}(P\rho)=\sum_{P}\tr^{4}(P\rho)=|G|+1
\ee
the above holds if $\rho=\frac{\bbbone}{d}+\sum_{P\in G}\phi_{P}P$. We just need to prove that $0<|G|<d-1$: as pointed out above, $\pur(\rho)=d^{-1}\sum_{P}\tr^{2}(P\rho)$ and since the purity of $\rho$ is bounded $d^{-1}\le\pur(\rho)\le1$ we have $0\le|G|\le d-1$. This concludes the proof.
\item The invariance under Clifford unitary operations follows from the unitary invariance of the $2-$R\'enyi entropy and from the fact that $M_{2}(\rho)$ is invariant under Clifford operations, see Sec. \ref{proofpropertiesalphamagic}.
\item Similarly, the additivity follows from the fact that $M_{2}(\rho\otimes \sigma)=M_{2}(\rho)+M_{2}(\sigma)$ and $S_{2}(\rho\otimes \sigma)=S_{2}(\rho)+S_{2}(\sigma)$, see Sec. \ref{proofpropertiesalphamagic}.
\end{enumerate}
 
\end{proof}


\section{Non-stabilizing power}
\subsection{Properties of the non-stabilizing power}\label{proofmagpowprop}
In this section, we give a proof of the properties $(i)$ and $(ii)$ for the non-stabilizing power presented in the main text.
\begin{enumerate}[label=(\roman*)]
    \item Let $C_{1},C_{2}$ be two Clifford operators. Then, thanks to the left/right invariance of the Haar measure over the Clifford group, we have $\mathcal{M}(UC_{2})=\mathcal{M}(U)$ and $\mathcal{M}(C_{1}UC_{2})=\mathcal{M}(C_{1}U)$. Then:
\be
\mathcal{M}(C_1 U)=\frac{1}{|\stab|}\sum_{\ket{\psi}\in \stab}M(C_{1}U\ket{\psi})=\frac{1}{|\stab|}\sum_{\ket{\psi}\in \stab}M(U\ket{\psi})=\mathcal{M}(U)
\ee
where we used the stability under Clifford operations of $M(\ket{\psi})$. 
\item For $U\in\mathcal{C}(d)$ the proof is trivial thanks to the left-right invariance of the Haar measure over the Clifford group. Conversely, if $C\not\in \mathcal{C}(d)$ there exists $\ket{\psi}\in\stab$ such that $U\ket{\psi}\not\in \stab$ and thus, thanks to the faithfulness of the non-stabilizerness measure $M(\ket{\psi})$ over pure states, one has $\mathcal{M}(U)>0$.
\subsection{Exact computation of the linear non-stabilizing power}\label{proofboundlrmag}
To compute the linear non-stabilizing power, we first calculate the average over the stabilizer states. Note that
\be 
\frac{1}{\stab}\sum_{\ket{\psi}\in\stab}\tr(QU^{\otimes 4}\st{\psi}^{\otimes 4}U^{\dag\otimes 4})=\mathbb{E}_{C}[\tr(QU^{\otimes 4} C^{\dag\otimes 4}\st{0}^{\otimes 4}C^{\otimes 4}U^{\otimes 4})]
\ee 
i.e. the average over stabilizer states can be seen as the average over the Clifford group $\mathbb{E}_{C}$ on a reference stabilizer state, e.g. $\st{0}^{\otimes4}$. The integration of $\st{0}^{\otimes 4}$ over the Clifford group\cite{zhu2016clifford,leone2021quantum} returns us:
\be
\mathbb{E}_{C}[\tr(QU^{\otimes 4} C^{\dag\otimes 4}\st{0}^{\otimes 4}C^{\otimes 4}U^{\otimes 4})]=\frac{1}{d(4+d)}\left(4+\frac{d}{\tr(Q\Pi_{\sym})}\tr(QU^{\otimes 4}Q\Pi_{\sym}U^{\dag\otimes 4}) \right)\label{avstate}
\ee
We can now look at the average of $M_{lin}(U)$, it reads:
\ba
\mathcal{M}_{lin}(U)&=&\frac{1}{|\stab|}\sum_{\psi\in\stab} M_{lin}(U\ket{\psi})=\frac{1}{|\stab|}\sum_{\psi\in\stab}\left(1- d\norm{\Xi(U\ket{\psi})}^2_{2}\right)\\
&=&1-\frac{d}{|\stab|}\sum_{\psi\in\stab}\tr(Q U^{\otimes 4}C^{\otimes 4}\st{0}^{\otimes 4}C^{\dag\otimes 4}U^{\dag\otimes 4} Q)\label{eq9}
=1-\frac{1}{4+d}\left(4+\frac{d}{\tr(Q\Pi_{\sym})}\tr(U^{\dag\otimes 4} QU^{\otimes 4}Q\Pi_{\sym})\right)\nonumber
\ea
where we used that $\stab^{-1}\sum_{\psi \in\stab}=1$, and for the last equivalence we used the results of Eq.\eqref{avstate}. This concludes the proof.

\subsection{Zero power unitaries}\label{proofmlin}
In this section, we prove the characterization for zero power unitaries, i.e. $\mathcal{M}_{lin}(U)=0$ iff $[U^{\otimes 4},Q\Pi_{\sym}]=0$. Starting from Eq.\eqref{eq9}, note that:
\be
\tr([U^{\otimes 4},Q\Pi_{\sym}]^{2})=2\tr(Q\Pi_{\sym})-2\tr(U^{\dag\otimes 4}QU^{\otimes 4}Q\Pi_{\sym})
\ee
Thus:
\be
\tr(U^{\dag\otimes 4}QU^{\otimes 4}Q\Pi_{\sym})=\tr(Q\Pi_{\sym})-\frac{1}{2}\tr([U^{\otimes 4},Q\Pi_{\sym}]^{2})
\ee
recalling that $\tr(A^2)\equiv \norm{A}_{2}^{2}$, we can rewrite the linear non-stabilizing power as:
\be
\mathcal{M}_{lin}(U)=\frac{d}{2(4+d)\tr(Q\Pi_\sym)}\norm{[U^{\otimes 4},Q\Pi_{\sym}]}_{2}^{2}
\ee
therefore we have $\mathcal{M}_{lin}(U)=0$ iff $\norm{[U^{\otimes 4},Q\Pi_{\sym}]}_{2}^{2}=0$ and this happens iff $[U^{\otimes4},Q\Pi_{\sym}]=0$. 
This concludes the proof. 
\end{enumerate}

\subsection{Lower bound to the $T-$count: proof of Eq.(12)}\label{tcount}
To prove the bound on the $T-$count, we first find the relationship between non-stabilizing power and the unitary stabilizer nullity defined in\cite{jiang2021lower} as:
\be 
v(U)=2n -\log_{2}s(U)
\ee 
where s(U) is defined as the number of $\pm1$ of $\tr(P_{1}U^\dag P_2U)/d$ varying $P_1,P_2$ in the Pauli group. Let $St(U)$ be the set of Pauli operators that returns $\tr(P_1U^\dag P_2U)/d=\pm 1$, we can bound $s(U)$ as:
\be
s(U)=\sum_{P_1,P_2\in St(U)}|\tr(P_1U^\dag P_2U)/d|\equiv\sum_{P_1,P_2\in St(U)}(\tr(P_1U^\dag P_2U)/d)^4\le \sum_{P_1,P_2\in\mathcal{P}_n}(\tr(P_1U^\dag P_2U)/d)^4
\label{chaininequality}
\ee
Now, starting from Eq.\eqref{c4}:
\be
\mathcal{M}_{lin}(U)=1-\frac{1}{4+d}\left(4+\frac{(d+3)}{4d^2}\sum_{P_{1},P_{2}}\mathbb{E}_{\ket{\psi}}[\langle\psi|U^{\dag}P_{1}UP_{2}|\psi\rangle^4]\right)
\ee
we can upper bound $\mathcal{M}_{lin}(U)$ using the Jensen inequality for the average on pure states and obtain:
\be
\mathcal{M}_{lin}(U)\le1-\frac{1}{4+d}\left(4+\frac{(d+3)}{4d^2}\sum_{P_{1},P_{2}}(\mathbb{E}_{\ket{\psi}}[\langle\psi|U^{\dag}P_{1}UP_{2}|\psi\rangle])^{4}\right)=1-\frac{1}{4+d}\left(4+\frac{(d+3)}{4d^2}\sum_{P_{1},P_{2}}(\tr(U^{\dag}P_{1}UP_{2})/d)^{4}\right)
\ee
using Eq.\eqref{chaininequality} we relate the linear non-stabilizing power to the unitary stabilizer nullity as:
\be
\mathcal{M}_{lin}(U)\le 1-\frac{1}{4+d}\left(4+\frac{(d+3)}{4d^2}\sum_{P_{1},P_{2}}(\tr(U^{\dag}P_{1}UP_{2})/d)^{4}\right)=1-\frac{1}{4+d}\left(4+\frac{(d+3)}{4}2^{-\nu(U)}\right)
\ee
The interesting feature of the unitary stabilizer nullity is that, being submultiplicative with respect to the product of unitaries, it lower bounds the \textit{$T$-counts}\cite{jiang2021lower} $t(U)$, the minimum number of $T$-gates needed in addition to Clifford circuits to decompose the unitary operator $U$:
\be
\nu(U)\le t(U)
\ee
we can further upper bound the linear non-stabilizing power with the $T-$counts $t(U)$ as:
\be 
\mathcal{M}_{lin}(U)\le1-\frac{1}{4+d}\left(4+\frac{(d+3)}{4}2^{-t(U)}\right)
\ee 
Let us express the above bound in terms of $t(U)$:
\be
t(U)\ge -\log_2\left(d-(4+d)\mathcal{M}_{lin}(U)\right)+\log_2(d+3)-2
\ee
this bound is obviously less tight than the one for the nullity\cite{jiang2021lower}.

\subsection{Average of non-stabilizing power over the unitary group: proof of Eq.(10)}\label{proofavermpowlin}
First of all, note that the average over all the stabilizer states can be treated as the average over the Clifford group on a reference stabilizer state. Since we need to compute the average over the full unitary group, we can exploit the left/right invariance of the Haar measure and avoid averaging over the Clifford group and instead average only over the full unitary group. The result reads:
\ba 
\mathbb{E}_{U}[\mathcal{M}_{lin}(U)]&=&\averu{M_{lin}(U\ket{\psi})}=1-d\averu{\tr(Q U^{\otimes 4}\st{\psi}^{\otimes 4}U^{\dag\otimes 4})}\\
&=&1-\frac{d\tr{Q\Pi_{\sym}}}{\tr{\Pi_{\sym}}}=1-\frac{4}{d+3}\nonumber
\ea
where we used the well known result\cite{leone2021quantum}:
\be
\mathbb{E}_{U}[U^{\otimes 4}\ket{\psi}\bra{\psi}^{\otimes 4}U^{\dag\otimes 4}]=\frac{\Pi_{\sym}}{\tr\Pi_{\sym}}
\ee
Then from the Jensen inequality follows:
\be
\averu{\mathcal{M}_{2}(U)}\ge-\log(1-\mathbb{E}_{U}[\mathcal{M}_{lin}(U)])
\ee

\subsection{Tipicality of non-stabilizing power: an application of Levy's lemma}\label{prooftipmpow}
Levy's lemma is frequently employed in quantum mechanics to prove strong concentration bounds on scalar functions of unitary operator $U\in\mathcal{U}(d)$ and of quantum pure states $\ket{\psi}$, see for example\cite{popescu2006entanglement, hayden2006aspects, Oliviero2020random}. The game consists in bounding the Lipschitz constant $\eta_{f}$ of a function $f$ (cfr.\cite{watrous2018theory}): if one finds $\lim_{d\rightarrow\infty}\frac{\eta_{f}^{2}}{d}=0$, then Levy's lemma proves a strong concentration of $f$ around its average value, for large $d$. Let us bound the Lipschitz constant of $\mathcal{M}_{lin}(U)$:
\ba
|\mathcal{M}_{lin}(U)-\mathcal{M}_{lin}(V)|&=&\frac{d}{(4+d)\tr(\Pi_{\sym}Q)}|\tr(Q\Pi_{\sym}(U^{\otimes 4}QU^{\otimes 4}-V^{\otimes 4}QV^{\otimes 4}))|\nonumber\\
&\le &\frac{d}{(4+d)\tr(\Pi_{\sym}Q)}\norm{Q\Pi_{\sym}}_{1}\norm{U^{\otimes 4}QU^{\otimes 4}-V^{\otimes 4}QV^{\otimes 4}}_{\infty}\nonumber\\
&\le&\frac{2d}{(4+d)}\norm{Q}_{\infty}\norm{U^{\otimes 4}-V^{\otimes 4}}_{\infty}\nonumber\\
&\le& \frac{8d}{(4+d)}\norm{U-V}_{\infty}\nonumber\\
&\le &\frac{8d}{(4+d)}\norm{U-V}_{2}
\ea
where in the first inequality we used the property of the Schatten $p-$norms: $|\tr(AB)|\le\norm{A}_{p}\norm{B}_{q}$ with $p^{-1}+q^{-1}=1$; in the second inequality we employed the following fact:
\ba
\norm{U^{\dag\otimes 4}QU^{\otimes 4}-V^{\dag\otimes 4}QV^{\otimes 4}}_{\infty}&=&\norm{U^{\dag\otimes 4}Q(U^{\otimes 4}-V^{\otimes 4})+(U^{\dag\otimes 4}-V^{\dag\otimes 4})QV^{\otimes 4}}_{\infty}\nonumber\\
&\le&\norm{QU^{\otimes 4}}_{\infty}\norm{U^{\otimes 4}-V^{\otimes 4}}_{\infty}+\norm{QV^{\otimes 4}}_{\infty}\norm{U^{\dag\otimes 4}-V^{\dag\otimes 4}}_{\infty}\nonumber\\
&\le& 2\norm{Q}_{\infty}\norm{U^{\dag\otimes 4}-V^{\dag\otimes 4}}_{\infty}
\ea
and the fact that $\norm{Q\Pi_{\sym}}_1=\tr(Q\Pi_{\sym})$. In the third inequality we used $\norm{Q}_{\infty}=1$ and $\norm{U^{\otimes 4}-V^{\otimes 4}}_{\infty}\le 4\norm{U-V}_{\infty}$ and lastly, in the fourth inequality, the hierarchy of the Schatten $p-$norms, namely $\norm{A}_{p}\le \norm{A}_{q}$ with $p>q$.

We thus find the Lipschitz constant to be bounded by a $O(1)$ constant, namely $\eta\le \frac{8d}{4+d}$ and the Levy's lemma applies (cfr.\cite{Oliviero2020random}).

\section{State Synthesis}
In this section we are going to give the formal statement of the theorem in the main text.
\begin{theorem*}
With overwhelming probability $p=1-\expf{-Cd^{1-2\beta}}$ where $\beta> 0$, a number $r\ge \frac{\beta n}{1-\log_{2}3/2}$ copies of the magic state $\ket{H}=\frac{1}{\sqrt{2}}(\ket{0}+e^{i\pi/4}\ket{1})$ are necessary to synthesize a $n-$qubit Haar random state.
\end{theorem*}
\begin{proof}
To synthesize a state $\ket{B}$ via stabilizer operations starting from a state $\ket{A}$, a necessary condition is that $M_{2}(\ket{A})\ge M_{2}(\ket{B})$ since the Clifford operations leave $M_2$ invariant. Here we prove that, with probability $p=1-\expf{-Cd^{1-2\beta}}$, the stabilizer $2-$R\'enyi entropy of a $n-$qubit random state $\ket{\psi}_{haar}$ is $M_{2}(\ket{\psi}_{haar})\ge \beta n$. On the other hand, from Eq. \eqref{magicstaterenyientropy}, we know that $M_{2}(\ket{H}^{\otimes r})=r(1-\log_{2}3/2)$ and thus, in order to have $M_{2}(\ket{H}^{\otimes r})\ge M_{2}(\ket{\psi}_{haar})$ one must (at least) have $M_{2}(\ket{H}^{\otimes r})\ge \beta n$. Thus one obtains a lower bound on the necessary copies of the magic state needed to synthesize a $n-$qubit Haar random state, i.e. $r\ge \frac{\beta n}{1-\log_{2}3/2}$. 

First of all, the average stabilizer linear entropy over the set of pure states reads:
\be
\bar{m}\equiv\mathbb{E}_{\ket{\psi}}[M_{lin}(\ket{\psi})]=1-\Theta(d^{-1})
\ee
see Eq. $(10)$. Then, from the typicality of the stabilizer linear entropy, picking a state $\ket{\psi}_{haar}$ sampled from the Haar distribution one has (see Eq. $(11)$):
\be
\operatorname{Pr}[|M_{lin}(\ket{\psi}_{haar})-\bar{m}|\ge \epsilon]\le \expf{-Cd\epsilon^{2}}
\ee
i.e. the stabilizer linear entropy of $\ket{\psi}_{haar}$ is $M_{lin}(\ket{\psi}_{haar})\ge 1-d^{-\beta}$ with probability $p=1-\expf{-Cd^{1-2\beta}}$. Setting $0<\beta<1/2$, we thus have that the probability is one up to an exponentially small error (in $d$). Consequently since $M_{2}(\ket{\psi}_{haar})=-\log (1-M_{lin}(\ket{\psi}_{haar}))$, the stabilizer $2-$R\'enyi entropy for a Haar random state:
\be
M_{2}(\ket{\psi}_{haar})\ge \beta n
\ee
with the same probability $p$. This concludes the proof.
\end{proof}

\section{non-stabilizing power and chaos}
\subsection{Average over the $k-$doped Clifford circuit}\label{proofkmpower}
The average over the $k-$doped Clifford circuit can be easily made by exploiting the techniques developed in\cite{leone2021quantum}:
\ba
\mathbb{E}_{\mathcal{C}_{k}}\left[\mathcal{M}_{lin}(U)\right]&=&1-a_{k}\tr{(QQ\Pi_{\sym})}-b_{k}\tr{(Q\Pi_{\sym})}\nonumber\\
&=&1-\frac{1}{3+d}\left(4+(d-1)\left(\frac{7d^2-3d +d(d + 3)\cos(4\theta)-8}{8(d^2-1)}\right)^{k}\right)\nonumber
\ea
where we used the average of $\st{\psi}^{\otimes 4}$ over the doped Clifford circuit (cfr. Application 3 of\cite{leone2021quantum}); for completeness, here we write down the coefficients $a_{k}$ and $b_{k}$ for a stabilizer state $\ket{\psi}\in\stab$ obtained in\cite{leone2021quantum}:
\ba
a_{k}&=&\frac{24 }{(d^2-1)(d+2)(d+4)}\left(\frac{(d+3)}{4}-1\right)\left(\frac{7d^2 - 3d +d(d + 3)\cos(4\theta)-8}{8(d^2-1)}\right)^{k}\\
b_{k}&=&\frac{1}{D_{\sym}}+\frac{24}{(d^2-1)(d+2)(d+4)}\left(\frac{4}{d(d+3)}-\frac{1}{d}\right)\left(\frac{7d^2 - 3d +d(d + 3)\cos(4\theta)-8}{8(d^2-1)}\right)^{k}\nonumber
\ea
\subsection{Relating the non-stabilizing power to the averaged fourth power of the $2-$OTOC: proof of Eq.(14)}\label{prooftheorem3}
Start from Eq. \eqref{eq9} and note that we can write $\Pi_{\sym}$ as:
\be
\Pi_{\sym}=\tr(\Pi_{\sym})\mathbb{E}_{\ket{\psi}} [\st{\psi}^{\otimes 4}]
\ee
where $\mathbb{E}_{\ket{\psi}}$ is the Haar measure over the set of pure states. Thus we rewrite Eq.(9) as:
\be
\mathcal{M}_{lin}(U)=1-\frac{1}{4+d}\left(4+\frac{d\tr(\Pi_{\sym})}{\tr(Q\Pi_{\sym})}\mathbb{E}_{\ket{\psi}}[\tr(U^{\dag\otimes 4} QU^{\otimes 4}Q\st{\psi}^{\otimes 4})]\right)
\ee
recalling that $Q\equiv d^{-2}\sum_{P\in\mathcal{P}_n}P^{\otimes 4}$, the above equation can be recast in the following form:
\be
\mathcal{M}_{lin}(U)=1-\frac{1}{4+d}\left(4+\frac{(d+3)}{4d^2}\sum_{P_{1},P_{2}}\mathbb{E}_{\ket{\psi}}[\langle\psi|U^{\dag}P_{1}UP_{2}|\psi\rangle^4]\right)
\label{c4}
\ee
finally $\aver{f(P)}_{P\in\mathcal{P}_n}\equiv\frac{1}{d^2}\sum_{P\in\mathcal{P}_n}f(P) $ is nothing but the average over the Pauli group $\mathcal{P}_n$ of the function $f(P)$.
\subsection{Tipicality of the $2-$OTOC with respect to $\ket{\psi}$: further application of Levy's lemma}\label{proofotoctyp}
In order to prove the above result we make use of the Levy's lemma\cite{watrous2018theory}. Let us compute the Lipschitz constant of the function $\otoc_{2}(\tilde{P}_1,P_2,\psi)\equiv\langle\psi|\tilde{P}_1P_{2}|\psi\rangle=\tr(\tilde{P}_1P_{2}\psi)$:
\ba
|\otoc_{2}(\tilde{P}_1,P_2,\psi)-\otoc_{2}(\tilde{P}_1,P_2,\phi)|&=&|\tr(\tilde{P}_1P_{2}(\psi-\phi))|\nonumber\\&\le&\norm{\tilde{P}_1P_2}_{\infty}\norm{\psi-\phi}_{1}\nonumber\\&=&2\sqrt{1-|\langle\psi|\phi\rangle|^2}\nonumber\\ &\le&2|\ket{\psi}-\ket{\phi}|
\ea
in the second equality we used the fact that $\norm{U}_{\infty}=1$ for any unitary operator. Thus we find that the Lipschitz constant $\eta\le 2$ and Levy's lemma applies, cfr.\cite{popescu2006entanglement}.
\subsection{The linear non-stabilizing power and the $8-$OTOC}\label{prooftheorem4}
Starting from Eq.(14) and using the fact that the $2-$point OTOC shows strong tipicality with respect to $\ket{\psi}$ we can write:
\be
\mathcal{M}_{lin}(U)\simeq 1-\frac{1}{4+d}\left(4+\frac{(d+3)d^2}{4}\Big\langle\left(\mathbb{E}_{\ket{\psi}}\left[\otoc_{2}(\tilde{P}_1,P_{2},\psi)\right]\right)^{4}\Big\rangle_{P_1,P_2}\right)
\ee
where $\simeq$ is an equality up to an exponentially small error in $d$. The Haar average of the $2-$point OTOC with respect to $\ket{\psi}$ reads:
\be
\mathbb{E}_{\ket{\psi}}\left[\otoc_{2}(\tilde{P}_1,P_{2},\psi)\right]=\mathcal{E}_{\ket{\psi}}\left[\tr(\tilde{P}_1P_{2}\psi)\right]=d^{-1}\tr(\tilde{P}_1P_{2})\equiv \otoc_{2}(\tilde{P}_1,P_2)
\ee
lastly, the following equation holds:
\be
\otoc_{8}(U)\equiv\langle d^{-1}\tr(\tilde{P}_1P_{2}P_3P_4\tilde{P}_{1}P_2P_4P_5\tilde{P}_1P_{2}P_5P_6\tilde{P}_1P_{2}P_6P_3)\rangle_{P_{3},P_{4},P_{5},P_{6}}=\otoc_{2}(\tilde{P}_1,P_{2})^{4}
\ee
where we used the fact that the average over the Pauli group reads:
\be
\aver{P_{3}\tilde{P}_1P_{2}P_{3}}_{P_3\in\mathcal{P}_n}=\frac{\bbbone}{d}\tr(\tilde{P}_1P_2)
\ee
\subsection{The average of the $8-$OTOC over the full unitary group}\label{aver8otoc}
In this section we compute the average over the full unitary group of the $8-$OTOC defined in the main text as:
\be
\aver{\otoc_8}\equiv \langle d^{-1}\tr(\tilde{P}_1P_{2}P_3P_4\tilde{P}_{1}P_2P_4P_5\tilde{P}_1P_{2}P_5P_6\tilde{P}_1P_{2}P_6P_3)\rangle_{P_{1},P_{2},P_{3},P_{4},P_{5},P_{6}}
\ee
In order to compute the average we recall the equality proved in the above section and
\be
\mathbb{E}_U[\aver{\otoc_8}]=\averu{\aver{\otoc_{2}(\tilde{P}_1,P_2)^{4}}_{P_1,P_2}}
\ee
by standard techniques we obtain:
\be
\mathbb{E}_U[\aver{\otoc_8}]=\frac{1}{d^4}\sum_{\pi\sigma}W_{\pi\sigma}\tr(Q T_{\pi})\tr(QT_{\sigma})=\frac{4(d^4-9d^2+6)}{d^6(d^2-9)}=\Theta(d^{-4})
\ee
where the projector $Q$ is defined as $Q=d^{-2}\sum_{P\in\mathcal{P}_n}P^{\otimes 4}$, $W_{\pi\sigma}$ are the Weingarten functions\cite{collins2003moments,collins2006integration} and $T_{\pi}$ are permutation operators corresponding to permutations $\pi\in \mathcal{S}_{4}$ the symmetric group of order $4$.
%